\documentclass[10pt,leqno]{amsart}
\usepackage{graphicx}
\usepackage{indentfirst,csquotes}
\usepackage{algorithm}
\usepackage{algorithmic}

\usepackage{amssymb,amsthm,amsmath}
\usepackage{xcolor,paralist,hyperref,titlesec,fancyhdr,etoolbox}
\usepackage{xfrac}
\usepackage{tikz}
\usepackage{orcidlink}
\hypersetup{colorlinks,linkcolor={blue},citecolor={blue},urlcolor={red}}  

\newtheorem{definition}{Definition}[section]
\newtheorem{example}{Example}[section]

\newtheorem{proposition}{Proposition}[section]

\newtheorem{problem}{Problem}[section]

\titleformat{\section}[hang]
  {\normalfont\large\bfseries\centering}
  {\thesection}
  {1em}
  {}

\hypersetup{ colorlinks=true, linkcolor=black, filecolor=black, urlcolor=black }

\usepackage{lipsum}

\begin{document}
\title{A numerical approach to the co-design of PID controllers and low-pass filters for time-delay systems} %%%%%%%%%%%% %%% Title should mention the codesign of filters and PID controllers
\author[Torres-García]{Diego Torres-García \orcidlink{0000-0003-3696-3450 }}\address{U2IS, ENSTA, Institut Polytechnique de Paris, 828 Boulevard des Maréchaux, 91120 Palaiseau, France.}
\author[Michiels]{Wim Michiels \orcidlink{0000-0002-0877-0080}}\address{Numerical Analysis and Applied Mathematics Center, Department of Computer Science, KU Leuven, Celestijnenlaan 200A, 3001 Heverlee, Belgium}
\date{\today}
\thanks{This work was supported by the project C14/22/092 of the Internal Funds KU Leuven}
\email{juan-diego.torres@ensta.fr, wim.michiels@kuleuven.be}
\maketitle
\let\thefootnote\relax
\begin{abstract}
    This paper addresses the numerical optimization of proportional–integral–derivative (PID) controllers for linear time-invariant systems with delays, where the derivative action is implemented using a low-pass filter. While performance assessment is often based on the spectral abscissa of the ideal PID-controlled system, the inclusion of a derivative filter fundamentally alters the closed-loop spectral properties and cannot be treated as a post-processing step. In particular, the spectral abscissa of the filtered closed-loop system may differ significantly from that of its unfiltered counterpart, potentially affecting both stability and performance.
    We propose a systematic numerical design framework in which the PID gains and the filter constant are optimized simultaneously by directly minimizing the spectral abscissa of the filtered closed-loop system. Treating the filter as an integral part of the control design allows us to reconcile robustness at high frequencies, in the sense of mitigating fragility issues due to approximate identities, with performance at low frequencies, in addition to counter measurement noise amplification.
    At the end of the presentation, numerical examples illustrate the proposed approach and highlight the benefits of controller–filter co-design. The results apply to general linear systems with input and/or state delays and are valid for both single-input single-output (SISO) and multi-input multi-output (MIMO) configurations.
\end{abstract} %%%%%%%%%

\section{Introduction}

    Proportional--integral--derivative (PID) controllers remain one of the most widely used feedback mechanisms in industrial and engineering applications because of their simplicity, transparency, and effectiveness across a broad range of processes \cite{Aastrom2021feedback, Ribeiro2020PID, ODwy2009PID}. In spite of their apparent simplicity, the rigorous analysis and design of PID controllers for systems with delays pose significant theoretical and practical challenges. Delay elements introduce an infinite-dimensional nature to the closed-loop system, thereby affecting stability, performance, and robustness in ways that cannot be fully captured by classical finite-dimensional design methods \cite{Niculescu2001, MichielsNiculescu2014}.
    
    A common practice in industrial implementations is to accompany the derivative action of the PID controller with a first-order low-pass filter. This modification is motivated by the need to reduce sensitivity to high-frequency noise, but it also has nontrivial consequences on the closed-loop dynamics. In particular, while the spectral abscissa associated with the ideal PID design is often adopted as a measure of stability and convergence speed, the inclusion of the filter alters the spectral distribution of the closed-loop characteristic roots. As a result, the spectral abscissa of the systems with filtered derivative action may deviate substantially from the \emph{nominal} one, leading to discrepancies between predicted and actual closed-loop behavior \cite{michiels2022filter}.
    
    Recent advances have emphasized the importance of explicitly accounting for delay effects in controller design. Spectral methods and quasi-polynomial root-location techniques have been developed for stability and performance optimization in delay differential equations \cite{Michiels2025relations}. In parallel, optimization-based approaches have been successfully applied to the tuning of PID controllers for delayed systems, often by minimizing the dominant characteristic root or by increasing robustness in the $H_{\infty}$ or $H_2$ norm sense \cite{Ramirez2015design, Torres2024stabilization, Michiels2023complete, Gomez2018optimization, Mhmood2025hinf}. However, most existing methods neglect the impact of the derivative filter on the spectral properties of the closed loop, implicitly assuming the idealized derivative action.
    
    The main idea of this work is that this discrepancy can be turned into an advantage if properly accounted for during the design phase. By explicitly integrating the \emph{filter} design with the optimization of the PID parameters, one can achieve improved stability properties while simultaneously guaranteeing other robustness constrains, such as a certain delay margin. This approach highlights that the filter is not merely a tool to avoid fragility of stability due to high-frequency perturbations (e.g., induced by so-called approximate identities neglected in the modeling phase \cite{Georgiou1989w}) when implementing the derivative action, but also a fundamental component that contributes to shaping nominal performance and robustness.
    
    The proposed framework applies to general linear time-invariant systems with input and/or state delays, and accommodates both single-input single-output (SISO) and multi-input multi-output (MIMO) configurations. In contrast to existing methods that optimize an unfiltered model and later add a filter heuristically, we directly optimize the \emph{filtered} spectral abscissa, including the filter constant as part of the optimization variables. This results in controllers that simultaneously reconcile performance, robustness, and realizability.
    
    The remainder of this paper is organized as follows. Section 2 introduces some of the key mathematical notions related to linear time-delay systems theory. Section 3 presents some of the classical problems that motivate this study. Section 4 formulates the problem statement, and the main results are presented in Section 5. Numerical examples complete this note in Section 6, and some conclusions and remarks are discussed in the last Section.

\section{Mathematical background}

%%%%%%%%%%%% LINEAR ONLY VERSION 

    Time-delay systems form a class of dynamical systems whose evolution depends not only on the current state but also on past values of the state. In this work, we consider linear time-invariant systems with possible multiple discrete constant delays in either the states or input/output channels, which can be represented in the general descriptor form:
    \begin{equation}
        \left\{ 
        \begin{array}{l}
        E \dot{x}(t) = A_0 x(t) +\sum_{j=1}^{N} A_j x\bigl(t-\tau_j\bigr) + B u(t), \\[10pt]
        y(t) = Cx(t)
        \end{array}
        \right.
    \label{eq:DDE_descriptor}
    \end{equation}
    where $x(t)\in\mathbb{R}^n$ is the state vector, $u(t)\in\mathbb{R}^m$ is the control input, $y(t) \in \mathbb{R}^m$ is the system's output at time $t$, $E,A_0,A_j\in\mathbb{R}^{n\times n}$, $B\in\mathbb{R}^{n\times m}$ and $C \in \mathbb{R}^{m\times n}$ are constant matrices, and $\tau_j \geq 0$, denote discrete, constant delays
    %, with $\tau_0$ representing an input/output channel delay
    . The matrix $E$ may be singular, allowing the representation of neutral-type systems through an equivalent descriptor formulation. Furthermore, let the columns of $U$ and $V$ define a basis for the left and wight null space of matrix $E$, respectively. Then, we assume in the rest of this paper that $U^TA_0 V$ is invertible and that either $U^T B = 0$ or $CV = 0$. These assumptions imply that \eqref{eq:DDE_descriptor} is \emph{semi-explicit}, and that there is no hidden feed-trough from $u$ to $y$.
    
    In contrast with ordinary differential equations, specifying an initial condition $x(t_0)$ is not sufficient to determine a unique solution of \eqref{eq:DDE_descriptor}. Instead, an initial function must be prescribed over the maximal delay interval
    $[-\tau_{\max},0]$, where $\tau_{\max} := \max\{\tau_1,\ldots,\tau_N\}$. More precisely, for a continuous input function, a forward solution of \eqref{eq:DDE_descriptor} is uniquely defined at $t_0$ for any initial condition $\phi \in \mathcal{X}_{t_0}$, where:
    \begin{equation}
        \mathcal{X}_{t_0} := \left\{ \phi \in \mathcal{C}\bigl([-\tau_{\max},0];\mathbb{R}^n\bigr): \quad U^T A_0 \phi(0) + \sum_{i=1}^N U^T A_i \phi(-\tau_i) + U^T Bu(t_0) = 0\right\}.
        \label{eq:consistency}
    \end{equation}
    Note that such structure guarantees that the algebraic or delay-difference equation is satisfied at $t = 0$, which guarantees continuity of the solution. 
    
    Accordingly, the state of the system at time $t$ can be identified with the function $x_t \in \mathcal{X}_t$ defined by:
    \[
        x_t : [-\tau_{\max},0] \to \mathbb{R}^n, \qquad x_t(\theta) := x(t+\theta).
    \]
    
    As a consequence of this functional formulation, time-delay systems form a class of infinite-dimensional systems. Although the state the trajectory $t \mapsto x(t)$ evolves in $\mathbb{R}^n$, the system state at each time instant is the function $x_t$. As a consequence, the stability properties of~\eqref{eq:DDE_descriptor} are governed by an infinite spectrum, and classical finite-dimensional stability tools are no longer directly applicable. Instead, spectral and semigroup-based techniques are required for stability analysis \cite{Breda2014stability,Sipahi2011stability}.

    One of the major challenges arising from the infinite-dimensional nature of time-delay systems lies in the computation and characterization of their spectrum. For linear time-invariant systems with delays, the characteristic equation takes the form of a \emph{quasi-polynomial}, whose set of roots determines the asymptotic behavior of the system. Unlike the finite-dimensional case, where the spectrum consists of a finite number of characteristic roots , the spectrum of a delay system contains \emph{infinitely many} characteristic roots. Consequently, the direct computation or assignment of all characteristic roots  is impossible. In the retarded case, that is, delays only appear on the state and not on its derivative, these characteristic roots  accumulate towards $-\infty$ in the complex plane. 
    
    In practice, one is primarily interested in the dominant part of the spectrum, namely the characteristic roots  with the largest real parts, since they dictate the decay rate and stability margins of the system. Several numerical tools have been developed to approximate the dominant characteristic roots within a prescribed region of the complex plane. Dedicated software packages, such as \textsc{TDS-Control} and \textsc{DDE-Biftool} \cite{Appeltans2023analysis,Engelborghs2002numerical}, provide reliable tools for this task. These packages allow for an accurate localization of the rightmost part of the spectrum, thereby offering a solid basis for the stability and performance analysis of delay systems.
    In this work, we employ the recent Python package \emph{tdcpy: Control of Time-Delay Systems in Python} \cite{tdcpy} to compute the characteristic roots throughout our study.
    
    Because the spectrum cannot be directly assigned through pole-placement techniques, as in the finite-dimensional case, more sophisticated design strategies are required to achieve desired closed-loop behavior. 
    A common performance-oriented objective is to \emph{maximize the decay rate} of the closed-loop solution, which is equivalent to placing the \emph{spectral abscissa}, the supremum of the real parts of the characteristic roots, as far to the left as possible. The spectral abscissa is defined as
    
    \begin{equation}
        \rho = \sup\{\,\Re(s) : \det(M(s)) = 0\,\},
        \label{eq:spectralAbscissa}
    \end{equation}
    where 
    \[
        M(s) = sE - A_0 - \sum_{j=1}^N A_j e^{-s \tau_j}.
    \]
    
    However, such optimization problems are highly nontrivial. Analytical results are available only for very low-order systems (typically of order two or less), while for higher-dimensional or multi-delay systems the results are fully numerically oriented. 
    Nevertheless, numerical optimization techniques based on spectral sensitivity analysis and continuation algorithms have proven effective in improving closed-loop performance and stability margins. 
    These approaches serve as the foundation for the filtered PID optimization framework developed in the present work.

\section{Effect of approximate identities on stability}
In the design of implementable PID controllers, a major challenge lies in the implementation of the derivative action. In practice, the derivative of the measured signal is rarely available directly, either due to sensor limitations or measurement noise. Consequently, it must be approximated, typically by introducing a low-pass filter or by employing a finite-difference scheme. From a frequency-domain viewpoint, such approximations introduce families of transfer matrices that converge to the identity as the approximation parameter tends to zero. These approximations, to which we may refer here as approximate identities, are often neglected, however, their impact on the stability properties of time-delay systems can be significant.

Indeed, when the open-loop system has relative degree one, its feedback interconnection with a PID controller is proper but not strictly proper. As a consequence, even if the resulting closed-loop system is exponentially stable, its stability may fail to be robust with respect to the \emph{infinitesimal} perturbations on parameters, induced by approximate identities. Note that this phenomenon led to the notion of $w$-stability in \cite{Georgiou1989w}, which guarantees a non-fragile feedback interconnection of plant and controller.

To formalize these ideas, we introduce the notion of an approximate identity, which provides a convenient framework to study families of operators that converge to the identity in the $H_\infty$ sense. This concept will allow us to analyze the limiting behavior of filtered control laws, delay difference approximations, and feedback delays, to characterize when stability properties are preserved, or lost.

\begin{definition}[Approximate identity \cite{Georgiou1989w}]
\label{def:identity_Hinf}
    Let $p\in\mathbb N$, and let
    \[
        L:\mathbb{C}\times\mathbb R_{+}\longrightarrow\mathbb C^{p\times p},
    \]
    we say that \(L\) is an approximate identity if:
    \[
        \lim_{r\to 0}\big\|L(\cdot\,;r)\big\|_{H_\infty} \;=\; 1.
    \]
    Furthermore, for any compact set $\Omega\subset \mathbb{C}$ and for any $\varepsilon>0$, there exists $\delta>0$ such that $0<r<\delta$ implies:
    \[
        \sup_{s\in\Omega}\big\|L(s;r)-I_p\big\|_2 < \varepsilon.
    \]
\end{definition}

% \begin{definition}[Approximate identity]
% A family $\{I_r\}_{r>0} \subset H_\infty^{p\times p}$ 
% is called an approximate identity if:

% \begin{enumerate}
% \item $\|I_r\|_{H_\infty} = 1$ for all $r>0$,

% \item for every compact set 
% $\Omega \subset \{ s\in\mathbb C : \Re s > 0 \}$,
% \[
% \sup_{s\in\Omega} \| I_r(s) - I_p \|_2
% \longrightarrow 0
% \quad \text{as } r \to 0.
% \]
% \end{enumerate}
% \end{definition}

Note that, in particular, the following correspond to approximate identities:
\begin{enumerate}
    \item $L_1(s; r) := e^{-s r}$
    \item $L_2(s; r) := \frac{1-e^{-s r}}{s r}$
    \item $L_3(s; r) := \frac{1}{s r + 1}$
\end{enumerate}
In the following, we aim to illustrate the effect on stability that an approximate identity may have, in a similar spirit to \cite{appeltans2022analysis, Mendez2024delay}.

\begin{example}[Small feedback delay]
We start by considering the following single-input single-output linear time-invariant system:
\begin{equation}
    \Sigma := \left\{ \begin{array}{l}
         \dot{x}(t) = \left[\begin{matrix}
             0 & 1\\ -3 & -4
         \end{matrix} \right]x(t) + \left[\begin{matrix}
             0\\1
         \end{matrix} \right] u(t)  \\[10pt]
          y(t) = [\begin{matrix}
              1 & -1
          \end{matrix}] x(t)
    \end{array} \right., \label{eq:sysEx1}
\end{equation}
with $u(t) = -k_py(t)-k_d\dot{y}(t)$ a classical PD controller. The transfer function associated to the system reads:
\begin{equation}
    H_{yu}(s) := C(sI-A)^{-1}B =  \frac{s-1}{(s+1)(s+3)},
\end{equation}
which leads to a closed-loop characteristic function $\Delta:\mathbb{C} \rightarrow \mathbb{C}$ given by:
\begin{equation}
    \Delta(s) = (1+k_d)s^2+(4+k_p-k_d)s+(3-k_p).
\end{equation}
It is well-known that for a linear system with no delays to be exponentially stable  all characteristic roots must be located on the open left-half plane $\mathbb{C}_-$. For a second order polynomial, it is also well-known that a sufficient and necessary condition for this is to have the same sign in all its coefficients, thus for this system to be exponentially stable we require $k_d > -1$, $k_p < 3$ and $k_p-k_d > -4$. These conditions are illustrated in Figure \ref{fig:stableE1}. Take, for instance, $k_p = 1$ and $k_d = 2$, placing the closed-loop characteristic roots at $s = (-3\pm \boldsymbol{i}\sqrt{15})/6$, guaranteeing the closed-loop system stability. A common feature of closed-loop control systems is the presence of delays in the input/output channels.

\begin{figure}
    \centering
    \includegraphics[width=0.5\linewidth]{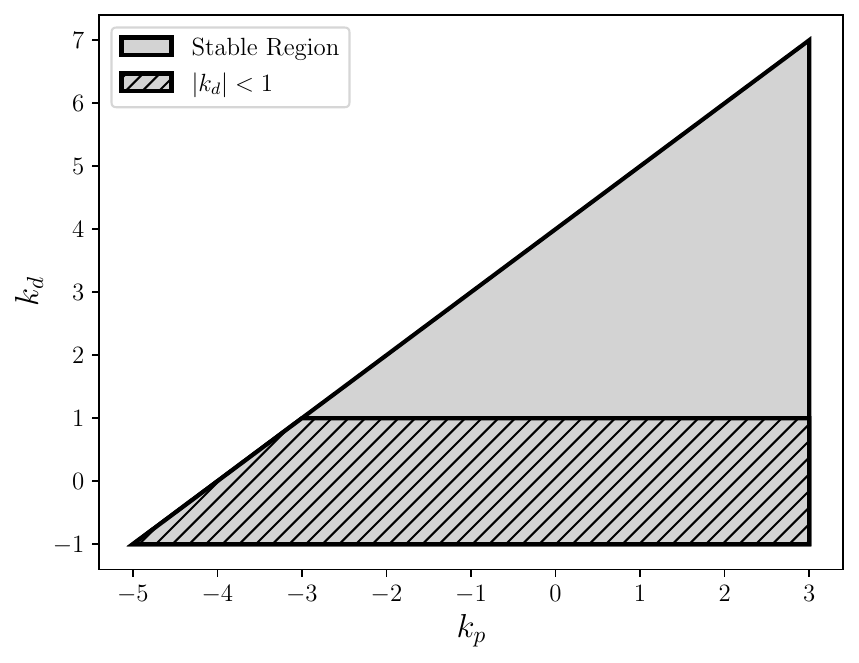}
    \caption{Stable region on the control parameters plane of system \eqref{eq:sysEx1}}
    \label{fig:stableE1}
\end{figure} 

On the presence of a delay $r\in\mathbb{R}_+$ (see Figure \ref{fig:Delay_ex1}), the characteristic equation rewrites:
\begin{equation}
    (1+k_de^{-rs})s^2+(4+(k_p-k_d)e^{-rs})s+(3-k_pe^{-rs}) = 0.
    \label{eq:neutralCharEq}
\end{equation}
Suppose $r\to0^+$, it is clear that $e^{-rs} \to 1$ uniformly on compact sets. Note that, with the presence of a delay, the term $(1+k_de^{-rs})$ corresponds to a difference equation, namely, $x(t)+k_dx(t-r)$, whose stability condition is given by $|k_d|<1$. With the given choice of gains, the associated difference equation is unstable in the presence of any $r>0$. We illustrate the effect of such a delay on the spectrum on Figure \ref{fig:neutralRoots}.

\begin{figure}
    \centering
    \tikzset{every picture/.style={line width=0.75pt}} %set default line width to 0.75pt        

\begin{tikzpicture}[x=0.75pt,y=0.75pt,yscale=-1,xscale=1]

%Rounded Rect [id:dp851930029248353] 
\draw   (167.46,112.04) .. controls (167.46,109.04) and (169.9,106.6) .. (172.9,106.6) -- (227.56,106.6) .. controls (230.56,106.6) and (233,109.04) .. (233,112.04) -- (233,128.36) .. controls (233,131.36) and (230.56,133.8) .. (227.56,133.8) -- (172.9,133.8) .. controls (169.9,133.8) and (167.46,131.36) .. (167.46,128.36) -- cycle ;
%Straight Lines [id:da1319327352561065] 
\draw    (233,120.4) -- (360.4,120.4) ;
\draw [shift={(363.4,120.1)}, rotate = 179.96] [fill={rgb, 255:red, 0; green, 0; blue, 0 }  ][line width=0.08]  [draw opacity=0] (8.93,-4.29) -- (0,0) -- (8.93,4.29) -- cycle    ;
%Straight Lines [id:da33271356627992255] 
\draw    (333.4,120.2) -- (333.4,184) ;
%Straight Lines [id:da5395862588450662] 
\draw    (63.8,184) -- (63.8,134) ;
\draw [shift={(64,131)}, rotate = 90.22] [fill={rgb, 255:red, 0; green, 0; blue, 0 }  ][line width=0.08]  [draw opacity=0] (8.93,-4.29) -- (0,0) -- (8.93,4.29) -- cycle    ;
%Shape: Circle [id:dp7321092450690665] 
\draw   (55.4,120.4) .. controls (55.4,115.65) and (59.25,111.8) .. (64,111.8) .. controls (68.75,111.8) and (72.6,115.65) .. (72.6,120.4) .. controls (72.6,125.15) and (68.75,129) .. (64,129) .. controls (59.25,129) and (55.4,125.15) .. (55.4,120.4) -- cycle ;
%Straight Lines [id:da7911712634231994] 
\draw    (9,120.4) -- (52.4,120.4) ;
\draw [shift={(55.4,121.4)}, rotate = 180.25] [fill={rgb, 255:red, 0; green, 0; blue, 0 }  ][line width=0.08]  [draw opacity=0] (8.93,-4.29) -- (0,0) -- (8.93,4.29) -- cycle    ;
%Straight Lines [id:da64246924922145] 
\draw    (72.6,120.4) -- (166.6,120.4) ;
%Rounded Rect [id:dp32628252633888577] 
\draw   (230.8,176.68) .. controls (230.8,173.98) and (232.98,171.8) .. (235.68,171.8) -- (258.12,171.8) .. controls (260.82,171.8) and (263,173.98) .. (263,176.68) -- (263,191.32) .. controls (263,194.02) and (260.82,196.2) .. (258.12,196.2) -- (235.68,196.2) .. controls (232.98,196.2) and (230.8,194.02) .. (230.8,191.32) -- cycle ;
%Rounded Rect [id:dp3742697622527076] 
\draw   (153.4,176.4) .. controls (153.4,173.64) and (155.64,171.4) .. (158.4,171.4) -- (199.4,171.4) .. controls (202.16,171.4) and (204.4,173.64) .. (204.4,176.4) -- (204.4,191.4) .. controls (204.4,194.16) and (202.16,196.4) .. (199.4,196.4) -- (158.4,196.4) .. controls (155.64,196.4) and (153.4,194.16) .. (153.4,191.4) -- cycle ;
%Shape: Rectangle [id:dp9516677725603723] 
\draw  [fill={rgb, 255:red, 0; green, 0; blue, 0 }  ,fill opacity=0.09 ][dash pattern={on 4.5pt off 4.5pt}] (147.8,161.4) -- (273,161.4) -- (273,203.8) -- (147.8,203.8) -- cycle ;
%Straight Lines [id:da6544561583215812] 
\draw    (333.4,184) -- (273.83,184) ;
%Straight Lines [id:da8818214455684856] 
\draw    (63.8,184) -- (146.6,184) ;
%Straight Lines [id:da6461003916404549] 
\draw    (204.6,184) -- (230.83,184) ;

% Text Node
\draw (313.35,102) node [anchor=north west][inner sep=0.75pt]    {$Y( s)$};
% Text Node
\draw (8.4,102) node [anchor=north west][inner sep=0.75pt]    {$R(s)$};
% Text Node
\draw (80,102) node [anchor=north west][inner sep=0.75pt]    {$E(s)$};
% Text Node
\draw (235.4,175.2) node [anchor=north west][inner sep=0.75pt]    {$e^{-rs}$};
% Text Node
\draw (115,165) node [anchor=north west][inner sep=0.75pt]    {$U(s)$};
% Text Node
\draw (162.8,175.3) node [anchor=north west][inner sep=0.75pt]    {$C(s)$};
% Text Node
\draw (177.6,110.8) node [anchor=north west][inner sep=0.75pt]    {$H_{yu}(s)$};

\end{tikzpicture}
    \caption{Presence of a delay on the input/output channel of the system.}
    \label{fig:Delay_ex1}
\end{figure}
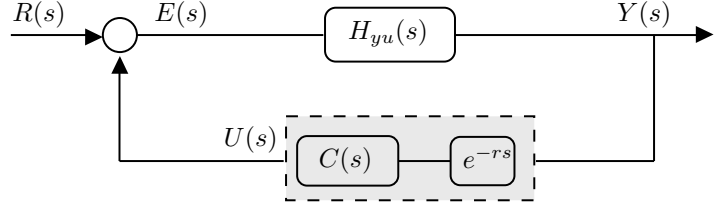

\begin{figure}
    \centering
    \includegraphics[width=0.5\linewidth]{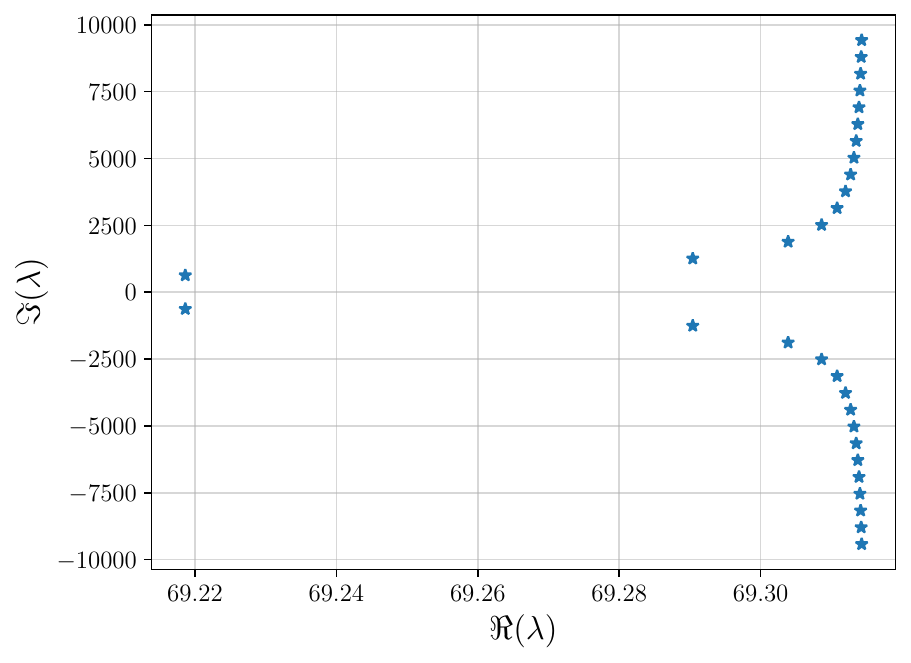}
    \caption{Unstable solutions of \eqref{eq:neutralCharEq} when $kp = 1$, $k_d=2$ and $r=0.01$.}
    \label{fig:neutralRoots}
\end{figure}

\end{example}

\begin{example}[Delay-difference approximations]
    
Next, we consider the following system:
\begin{equation}
    \Sigma := \left\{ \begin{array}{l}
         \dot{x}(t) = \left[\begin{matrix}
             -1 & 1 & 1\\ 1 & 0 & 0\\ 0 &1 &0
         \end{matrix} \right]x(t) + \left[\begin{matrix}
             -5\\0\\0
         \end{matrix} \right] u(t)  \\[20pt]
          y(t) = [\begin{matrix}
              1 & 0 & 0
          \end{matrix}] x(t)
    \end{array} \right.,
\end{equation}
which has a corresponding closed-loop characteristic function $\Delta: \mathbb{C} \rightarrow \mathbb{C}$ given by:
\begin{equation}
    \Delta(s) = s^3+s^2-s-1-5s^2(k_p+k_ds).
\end{equation}
Let us consider $k_p = 3$ and $k_d = 2$, which exponentially stabilize the system by placing the closed-loop characteristic roots at $s = -1.53$ and $s=-0.012 \pm 0.269 \boldsymbol{i}$. We propose replacing the derivative by a finite difference approximation, namely:
\[
    s \approx \frac{1-e^{-rs}}{r}.
\]
Note that this approximation can be rewritten as the inclusion of an approximate identity by:
\[
    \frac{1-e^{-rs}}{r} = s\left(\frac{1-e^{-rs}}{rs}\right).
\]
The characteristic equation then rewrites:
\begin{equation}
    s^3 + s^2 - s - 1 - 5s^2 \left(k_p +k_d \left(\frac{1-e^{-rs}}{r}\right)\right) = 0,
\end{equation}
introducing the scaling variable $\bar{s} = sr$, and multiplying by $r$ we can obtain the following expression:
\[
    \bar{s}^3+r\bar{s}^2-r^2\bar{s}-1-5\bar{s}^2(k_pr+k_d(1-e^{-\bar{s}})) = 0.
\]
Note that by taking $r \to 0^+$, we should expect the approximation of the derivative to improve, however, observe that with $r\to 0^+$, the characteristic function converges on compact sets to:
\begin{equation}
    \bar{s}^2\underbrace{(\bar{s}-5k_d(1-e^{-\bar{s}}))}_{f(\bar{s};k_d)}=0.
    \label{eq:ddaproxEx}
\end{equation}
From Figure \ref{fig:ddaprox} it is clear that there exists a positive zero of \eqref{eq:ddaproxEx}, and by Rouché's theorem, there is a characteristic root converging to this zero as $r\rightarrow0^+$. Moreover, after rescaling such a root goes to $+\infty$ as $r\to0^+$. Similarly to the input/output delay case (see, for instance, \cite{Hale2002Strong, Hale2003Stability}.), this situation restricts the choice of the derivative gains or matrices, and illustrates the importance of properly choosing them. 

\begin{figure}
    \centering
    \includegraphics[width=0.5\linewidth]{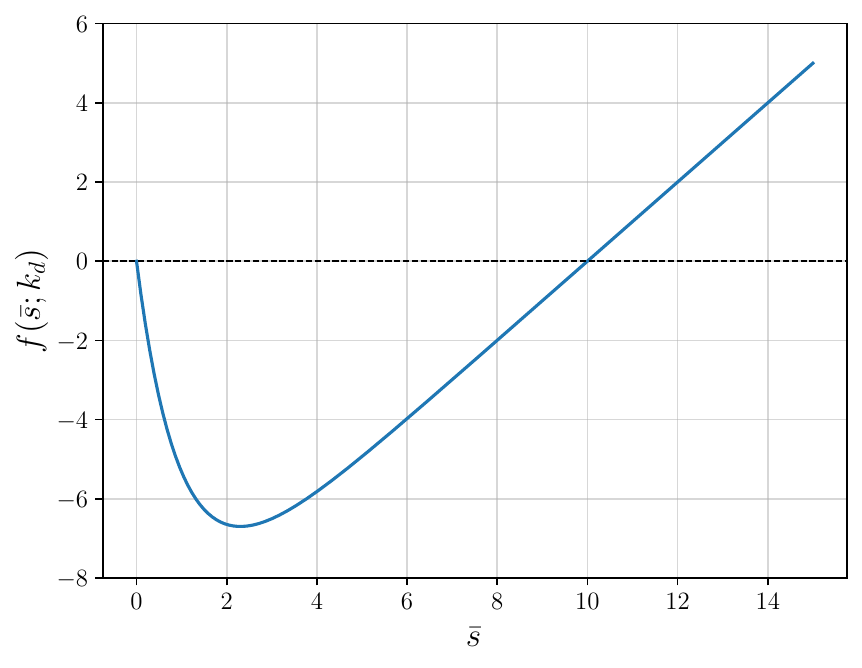}
    \caption{Evolution of $\bar{s}\mapsto f(\bar{s};k_d)$ with $k_d = 2$.}
    \label{fig:ddaprox}
\end{figure}

\end{example}

\begin{example}[Choice of filter constant]
Finally, consider the following system: 
\begin{equation}
    \Sigma := \left\{ \begin{array}{l}
         \dot{x}(t) = \left[\begin{matrix}
             0 & 5\\ 1 & 0
         \end{matrix} \right]x(t) + \left[\begin{matrix}
             1\\-1
         \end{matrix} \right] u(t)  \\[10pt]
          y(t) = [\begin{matrix}
              1 & 0
          \end{matrix}] x(t)
    \end{array} \right.,
    \label{eq:exDelayMargMot}
\end{equation}
The lack of $w$-stability implies that stability can be lost even for arbitrarily small values of $r$, as illustrated with the previous example. An intuitive solution to make the feedback interconnection strictly proper and solve these fragility problems is to add a low-pass filter to the controller, on the condition that the filter itself does not destabilize the system. This condition is not trivially satisfied (note that the $1/(rs+1)$ is also an approximate identity). In this spirit, let us once again consider a PD controller, but now with the inclusion of a low-pass filter to the derivative action, with cut-off frequency $\omega_c = 1/r$, leading to the following characteristic equation:

\begin{equation}
    s^2-5+(5-s)\left(k_p+\frac{k_ds}{rs+1}\right)=0.
\end{equation}

A similar issue as the one presented with the inclusion of the approximation of the derivative action arises here, and stability can be destroyed if the filter acts as a destabilizing perturbation as discussed in \cite{michiels2022filter}. Let us assume that is not the case, and we choose the value of $k_d$ such that stability is preserved for sufficiently small $r$. Take, for example, $k_d = 1/2$ and $k_p = 2$. 

Two issues are to be faced here. First, with the ideal PD controller, the dominant roots of the closed-loop system are located at $s = -\tfrac{1}{2}(1\pm \sqrt{39}\boldsymbol{i})$. After inclusion of the filter, the spectra of the system would be necessarily changed, even if stability is preserved. If, for example, we take $r = 0.01$, an additional root corresponding to $s=-47.057$ appears. Naturally, the smaller the value of $r$, the more to the left the additional root will move, making the closed-loop system closer to the nominal one. Stability is now guaranteed for small enough values of $r$, and for sufficiently small values of $r$, we also know a finite delay margin exists. However, we have no information on how big these two margins are. Indeed, as depicted in Figure \ref{fig:delayMargMot}, both margins are very small, and a small perturbation may destroy the stability of the system.

\begin{figure}
    \centering
    \includegraphics[width=0.5\linewidth]{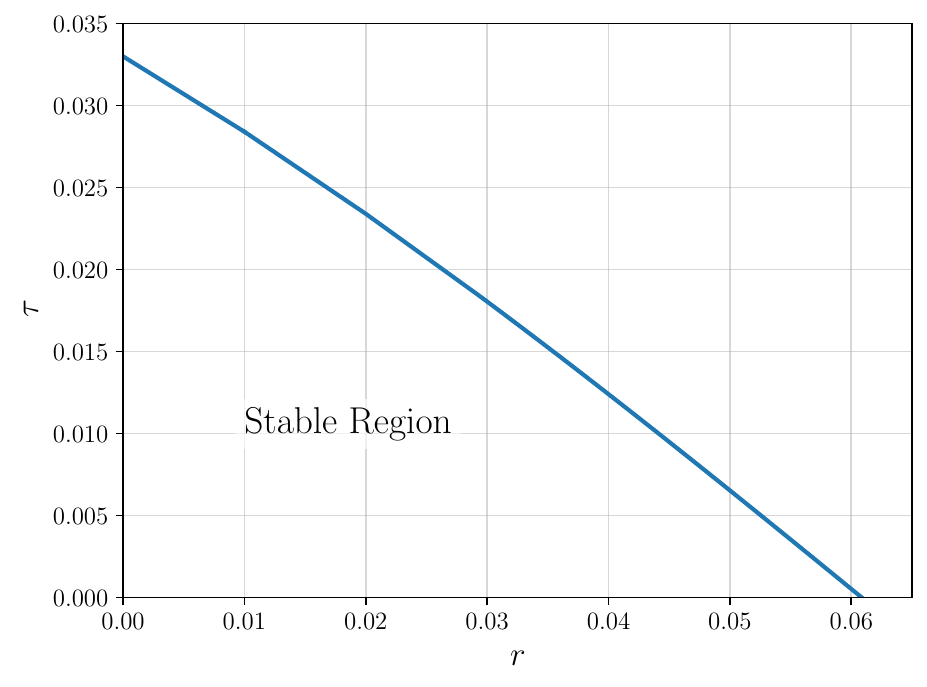}
    \caption{Stability region in the $(r, \tau)$-plane for $\tau$ a delay in the input channel, and $r$ the filter constant from the closed-loop between system \eqref{eq:exDelayMargMot} and a PD controller}
    \label{fig:delayMargMot}
\end{figure}

\end{example}

The issue related to safely including a low-pass filter has been deeply discussed in the literature. This has conducted to the notion of filtered spectral abscissa, which corresponds to the limit superior of the spectral abscissa when the cut-off frequency of an inserted filter goes to infinity. Thus, if the filtered spectral abscissa is negative, a filter with \emph{sufficiently high} cut-off frequency can be inserted while preserving stability. This results in $w$-stability and nonzero stability margins for any approximate identity, but no guarantees are given in terms of the size of such margins.

The main idea of works such as \cite{appeltans2022analysis} is to design strongly stabilizing controllers by optimizing the spectral abscissa, while constraining the control gains to those for which the filtered spectral abscissa is strictly negative. However, this approach has a main disadvantage, as we only guarantee that stability is not lost for sufficiently high cut-off frequency, that is, the system delay margin is not zero. However, as illustrated in the previous example, since the filter is not included in the control design, there are no guarantees in terms of robustness or performance, apart from the delay margin being finite. In the rest of this note we aim to propose a methodology such that the filter design is included together with the control itself.

\section{Problem statement}
    On this paper we consider LTI systems with discrete delays on the state and input/output channel of the form:
    \begin{equation}
        \Sigma :=\left\{
        \begin{array}{l}
            \dot{x}(t) = A_0 x(t) + \sum_{i = 1}^N A_i x(t-\tau_i) + B u(t-\tau_0) \\
            y(t) = C x(t)\end{array}
        \right.,
    \end{equation}
    where $x \in \mathbb{R}^n$ is the state variable, $u \in \mathbb{R}^m$ is the input, $y \in \mathbb{R}^p$ is the output of the system, $\tau_0 \in \mathbb{R}_+$ is the input delay, $\tau_1 < \tau_2 < \dots < +\infty$ are discrete state delays, $A_k \in \mathbb{R}^{n\times n}$ for $k \in \{0, 1, \dots,N\}$, $B\in \mathbb{R}^{n \times m}$, and $C \in \mathbb{R}^{p \times n}$. We also consider a classical PID control law, including a low-pass filter on the derivative action, given by:
    \begin{equation}
        \left\{ \begin{array}{l}u(t) = K_p y(t) + K_d z(t) + K_i \int_0^t y(s) ds\\
        T \dot{z}(t) + z(t) = \dot{y}(t) \end{array} \right.,
    \end{equation}
    where $K_p$, $K_i$ and $K_d$ are $m \times p$ real valued matrices, and $T>0$ is the filter constant, and $\omega_c=1/T$ is the cut-off frequency associated to the filter. 
    % \begin{assumption}
    % \label{asumption:nominal stab}
    %     In the absence of input/output channels delays (i.e., for $\tau_0 = 0$), it is assumed that there exists a set of matrices (gains) $(K_p, K_i, K_d)$ that stabilizes the nominal delay-free closed-loop system in the absence of the filter (i.e., for $T=0$).
    % \end{assumption}
    
   The closed-loop dynamics can be equivalently represented in the form of a delay differential--algebraic equation (DDAE). It is worth noting that, thanks to the inclusion of the low-pass filter, the closed-loop system is of the retarded type. We introduce the auxiliary variable $w(t)$:
    \[
        \dot{w}(t) = y(t),
    \]
    and define the extended state vector:
    \[
        \xi(t) = \begin{bmatrix} x(t)^{\top} & w(t)^{\top} & z(t)^{\top} \end{bmatrix}^{\top}.
    \]
    With this notation, the closed-loop system can be written as:
    \begin{equation}
    \begin{aligned}
        \dot{\xi}(t) 
        &- 
        \begin{bmatrix}
            0 & 0 & 0 \\
            0 & 0 & 0 \\
            T^{-1} C & 0 & 0
        \end{bmatrix} 
        \dot{\xi}(t) \\
        &=
        \begin{bmatrix}
            A & 0 & 0 \\
            C & 0 & 0 \\
            0 & 0 & -T^{-1}I
        \end{bmatrix} \xi(t) 
        + 
        \begin{bmatrix}
            B K_p C & B K_i & B K_d \\
            0 & 0 & 0 \\
            0 & 0 & 0
        \end{bmatrix} \xi(t-\tau_0)  \\
        &\quad + 
        \sum_{k=1}^m
        \begin{bmatrix}
            A_k & 0 & 0 \\
            0 & 0 & 0 \\
            0 & 0 & 0
        \end{bmatrix} \xi(t-\tau_k).
    \end{aligned}
    \end{equation}
    
    The stability properties of this DDAE are governed by the associated characteristic equation. 
    For compactness, define the block matrices:
    \begin{equation*}
        E = I - 
        \begin{bmatrix}
            0 & 0 & 0 \\
            0 & 0 & 0 \\
            T^{-1} C & 0 & 0
        \end{bmatrix}, 
        \quad
        M_0 =
        \begin{bmatrix}
            A & 0 & 0 \\
            C & 0 & 0 \\
            0 & 0 & -T^{-1} I
        \end{bmatrix},
    \end{equation*}
    \begin{equation*}
        M_1 =
        \begin{bmatrix}
            B K_p C & B K_i & B K_d \\
            0 & 0 & 0 \\
            0 & 0 & 0
        \end{bmatrix}.
    \end{equation*}
    
    The closed-loop characteristic matrix is then given by:
    \begin{equation}
        M(s) 
        = s E - M_0 - M_1 e^{-\tau_0 s}
        - \sum_{k=1}^m 
        \begin{bmatrix}
            A_k & 0 & 0 \\
            0 & 0 & 0 \\
            0 & 0 & 0
        \end{bmatrix} e^{-\tau_k s},
        \label{eq:charMatrix}
    \end{equation}
    and the characteristic equation is given by:
    \begin{equation}
        \det\!\big(M(s)\big) = 0.
    \end{equation}
    
    As in the finite-dimensional case, exponential stability is determined by the spectral abscissa function $(K_p, K_i, K_d, T) \mapsto \rho(K_p, K_i, K_d, T)$ of the closed-loop system, defined as:
    \begin{equation}
        \rho(K_p, K_i, K_d, T)  = \max_{s \in \mathbb{C}} \Big\{ \Re(s) : \det\!\big(M(s)\big) = 0 \Big\}.
    \end{equation}
    Thanks to the inclusion of the filter, the system is of the retarded type, thus, the closed-loop system is exponentially stable if and only if the spectral abscissa satisfies $\rho < 0$. In order to incorporate practical constraints on the derivative filter, we seek to optimize the closed-loop spectral abscissa while restricting the filter parameter to a prescribed interval $T \in [T_{\min},T_{\max}]$. 
    The resulting constrained optimization problem can be formulated as
    \begin{equation}
    \begin{array}{cl}
    \underset{K_p,K_i,K_d,T}{\text{minimize}} & \rho(K_p,K_i,K_d,T) \\[1mm]
    \text{subject to} & T_{\min} \leq T \leq T_{\max}.
    \end{array}
    \end{equation}
    
    \subsection*{Motivating example} 
    Let us consider the following motivating example, adapted from Appeltans et al.~\cite{appeltans2022analysis}. Consider the system
    \begin{equation}
        \begin{matrix}
            & \dot{x}(t) = \left[\begin{matrix}
                -1 & \tfrac{1}{3} & 1\\ 
                1 & 0 & 0 \\ 
                0 & 1 & 0
            \end{matrix} \right] x(t)
            + 
            \left[\begin{matrix}
                2\\0\\0
            \end{matrix}\right] u(t-\tau_0) \\[20pt]
            & y(t) = 
            \left[\begin{matrix}
                0.5 & 0 & 0.5
            \end{matrix}\right] x(t).
        \end{matrix}
    \end{equation}

    Assume a PD controller of the form
    \[
        u(t) = k_p y(t) + k_d \dot{y}(t).
    \]
    The objective is to design the controller so as to optimize the spectral abscissa of the closed-loop system when the derivative action is implemented through a low-pass filter.

    In \cite{appeltans2022analysis}, the optimization is carried out for the delay-free system, under additional constraints ensuring that the filtered spectral abscissa is strictly negative. The resulting locally optimal gains are $k_p = -1.08015$ and $k_d = -1.04045$.

    Without filtering, the closed-loop system is neutral. In particular, since $|k_d| > 1$, stability is lost for any input delay $\tau_0 > 0$. This makes the inclusion of a low-pass filter necessary.

    Figure~\ref{fig:stableRegion} depicts the stability crossing curves in the $(T,\tau_0)$-plane. It can be observed that the delay margin, that is, the maximum admissible input delay $\tau_0$ preserving stability, depends strongly on the choice of the filter constant $T$. Moreover, for very small values of $T$, high-frequency modes may arise, leading to increased sensitivity to measurement noise.

    \begin{figure}
        \centering
        \includegraphics[width=0.5\linewidth]{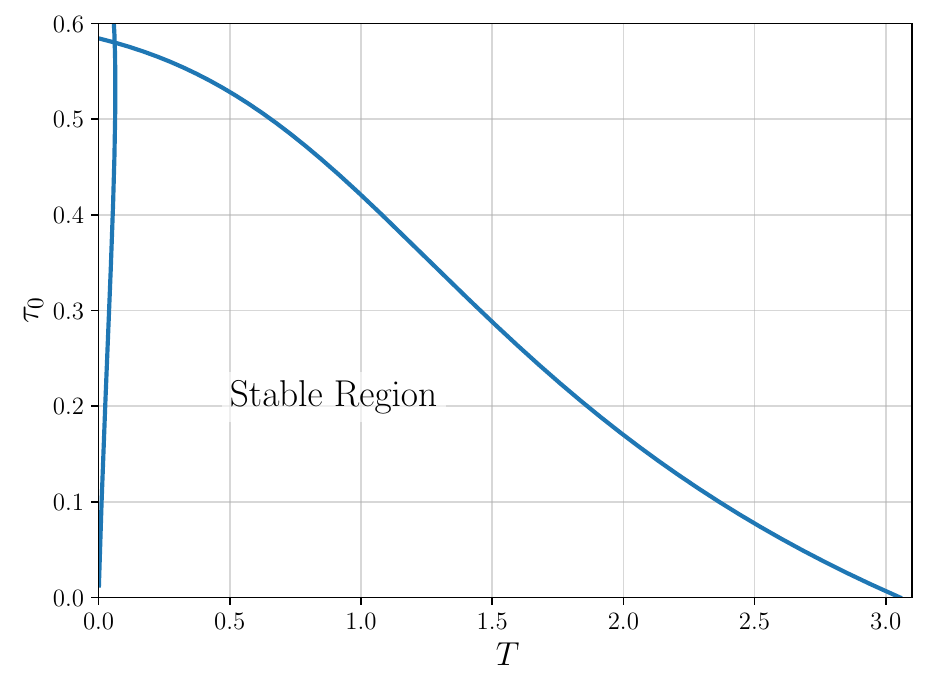}
        \caption{Stability crossing curves in the $(T,\tau_0)$-plane.}
        \label{fig:stableRegion}
    \end{figure}

    After introducing the filter, the closed-loop characteristic function $\Delta:\mathbb{C}\rightarrow\mathbb{C}$ becomes
    \begin{equation}
        \Delta(s)
        =
        (sT+1)(s^3 + s^2 -\tfrac{1}{3}s - 1 )
        -
        (s^2+1)\big((Tk_p+k_d)s + k_p\big).
    \end{equation}

    To illustrate potential applications of incorporating the filter into the design process, suppose that the system is required to tolerate at least a delay of $\tau_{\text{min}} = 0.2$. From Figure~\ref{fig:stableRegion}, and using the gains proposed in \cite{appeltans2022analysis}, this requirement imposes approximately $T \leq 1.75$.

    For instance, fixing $\tau_0 = 0.2$ and choosing $T = 0.1$, the spectral abscissa of the filtered closed-loop system with the gains of \cite{appeltans2022analysis} is $\rho = -0.1475$. If instead the three parameters $(k_p, k_d, T)$ are optimized jointly by minimizing the spectral abscissa subject to $T \in (0,1.75]$, a local minimizer is obtained at
    \[
        k_p = -1.0979, 
        \quad 
        k_d = -1.2259,
        \quad 
        T = 0.1740.
    \]

    It is important to note that modifying $k_p$ and $k_d$ alters the stability region shown in Figure~\ref{fig:stableRegion}. Consequently, the admissible interval for $T$ becomes gain-dependent. A natural strategy is therefore to solve the optimization problem, recompute the delay margin for the updated gains, and iterate with revised bounds on $T$. This idea forms the basis of the procedure proposed later in this work.

    With the jointly optimized parameters, the spectral abscissa improves to $\rho = -0.2435$. The rightmost characteristic roots for both configurations are shown in Figure~\ref{fig:rootsExample}.

    \begin{figure}
        \centering
        \includegraphics[width=0.5\linewidth]{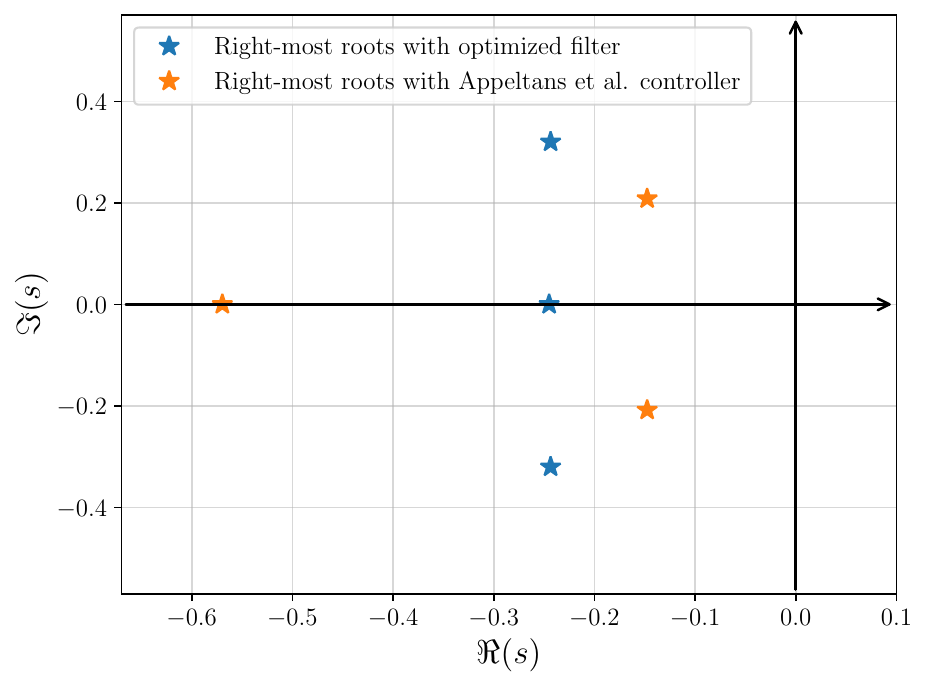}
        \caption{Comparison of the rightmost closed-loop roots obtained with fixed filtering and with joint controller–filter optimization.}
        \label{fig:rootsExample}
    \end{figure}

    This example highlights the importance of treating the filter constant as a control design parameter. Motivated by these observations, we propose a numerical procedure for the joint design of the PID gains and the filter constant, allowing the spectral abscissa to be optimized while selecting $T$ in accordance with the specific requirements and constraints of the problem.

\section{Filtered PID control design}

Motivated by the previous discussion, we formulate the design of a filtered PID controller as a spectral-abscissa minimization problem in which the filter constant is treated as a design variable. In particular, we seek to optimize the closed-loop spectral abscissa while restricting the filter constant to a prescribed admissible interval.

\begin{problem}[Spectral-abscissa minimization]
\label{prob:opt}
    \[
        \begin{array}{ll}
        \underset{K_p, K_i, K_d, T}{\text{minimize}} 
        & \rho(K_p, K_i, K_d, T) \\
        \text{subject to} 
        & T_{\min} < T < T_{\max}.
        \end{array}
    \]
\end{problem}

    To handle the constraints on $T$, we reformulate this constrained problem as an unconstrained one by introducing a penalty function. More precisely, we define the penalized spectral abscissa
    \begin{equation}
    \tilde{\rho}(K_p,K_i,K_d,T) :=
    \begin{cases}
    \rho(K_p,K_i,K_d,T), & T_{\min} \leq T \leq T_{\max}, \\[6pt]
    \rho(K_p,K_i,K_d,T_{\min}) + \alpha (T_{\min}-T), & T < T_{\min}, \\[6pt]
    \rho(K_p,K_i,K_d,T_{\max}) + \alpha (T-T_{\max}), & T > T_{\max},
    \end{cases}
    \label{eq:penalized_rho}
    \end{equation}
    where $\alpha > 0$ is a fixed penalty coefficient.
    With this construction, minimizing $\tilde{\rho}$ guarantees obtaining a feasible solution for the constrained problem. 
    %As a result, the optimized controller simultaneously achieves improved nominal performance and guarantees the prescribed level of robustness.
    
    It is worth noting the adopted kind of penalty functions are often discouraged due to the loss of smoothness they introduce. In the present setting, however, this drawback is not restrictive, since the spectral abscissa function is itself non-smooth. Consequently, non-smooth optimization techniques are required regardless of the penalty, making the proposed approach well suited to the problem at hand. On this work we use BFGS on SciPy \cite{2020SciPy} and HANSO as solvers. 
    
    The following result establishes that local minimizers of the penalized objective coincide with local minimizers of the constrained problem.

\begin{proposition}[Local optimality]
\label{prop:optimality}
    Let $\tilde{\rho}(K_p,K_i,K_d,T)$ denote the penalized spectral abscissa defined in \eqref{eq:penalized_rho}. If $(K_p^*,K_i^*,K_d^*,T^*)$ is a local minimizer of $\tilde{\rho}$, then it is also a local minimizer of Problem~\ref{prob:opt}.
\end{proposition}

\begin{proof}
For fixed controller gains, the function 
$T \mapsto \tilde{\rho}(K_p,K_i,K_d,T)$
is strictly larger outside the interval $[T_{\min},T_{\max}]$ than at the corresponding boundary points. Hence, no point with $T < T_{\min}$ or $T > T_{\max}$ can be a local minimizer of $\tilde{\rho}$, since moving $T$ toward the nearest endpoint strictly decreases the objective.

Let $(K_p^*,K_i^*,K_d^*,T^*)$ be a local minimizer of $\tilde{\rho}$. The previous observation implies that $T^* \in [T_{\min},T_{\max}]$. On this feasible interval, the penalized objective coincides with the original spectral abscissa, i.e.,
\[
\tilde{\rho}(K_p,K_i,K_d,T)=\rho(K_p,K_i,K_d,T),
\quad
T \in [T_{\min},T_{\max}].
\]
Therefore, $(K_p^*,K_i^*,K_d^*,T^*)$ is a local minimizer of $\rho$ restricted to the feasible set, which completes the proof.
\end{proof}

\subsection*{Control design procedure}

The proposed co-design strategy for the PID gains and the filter constant is summarized in Algorithm~\ref{alg:filteredPID}. The procedure combines nonsmooth spectral-abscissa minimization with a practical initialization strategy ensuring closed-loop stability.

\begin{algorithm}
\caption{Filtered PID control design}
\label{alg:filteredPID}
\begin{algorithmic}[1]
\REQUIRE System matrices and delays $(A_0,A_1,\dots,A_N,B,C,\tau_1,\dots,\tau_N)$,
initial strongly stabilizing gains $(K_p^{(0)},K_i^{(0)},K_d^{(0)})$,
penalty parameter $\alpha>0$
\ENSURE Optimal filtered PID gains $(K_p^\star,K_i^\star,K_d^\star)$ and filter constant $T^\star$

\STATE \textbf{Initialization:}
\STATE Choose $(K_p^{(0)},K_i^{(0)},K_d^{(0)})$ such that the closed-loop system is strongly stable.
\STATE Select an admissible interval $[T_{\min},T_{\max}]$ for the filter constant.
\STATE Choose $T^{(0)}\in[T_{\min},T_{\max}]$.

%\STATE \textbf{Iterative optimization:}
%\FOR{$k=0,1,2,\dots$ until
% $\displaystyle
% \frac{\big|\tilde{\rho}^{(k+1)}-\tilde{\rho}^{(k)}\big|}
% {1+\big|\tilde{\rho}^{(k)}\big|}
% < \varepsilon$}
    \STATE Solve a local minimization problem
    \[
        \underset{K_p,K_i,K_d,T}{\text{minimize}}
        \quad
        \tilde{\rho}(K_p,K_i,K_d,T)
    \]
    % initialized at $(K_p^{(k)},K_i^{(k)},K_d^{(k)},T^{(k)})$.
    % \STATE Let $(K_p^{(k+1)},K_i^{(k+1)},K_d^{(k+1)},T^{(k+1)})$
    % denote the resulting local minimizer.
%\ENDFOR

\STATE \textbf{Output:}
$(K_p^\star,K_i^\star,K_d^\star,T^\star)$
\end{algorithmic}
\end{algorithm}

In the next section, we illustrate the proposed co-design methodology through a numerical example, highlighting how simultaneous optimization of the PID gains and the filter constant can improve both spectral performance and system robustness.

\section{Numerical examples}

In this section, we present two numerical examples to illustrate the proposed filtered PID control design framework. The first example aims to illustrate the importance of considering the filter as part of the control design when working with derivative actions, while the second example demonstrates how the interval $[T_\text{min}, T_\text{max}]$ in Algorithm \ref{alg:filteredPID} can be adapted in an outer iteration to arrive to a good compromise between performance, in terms of spectral abscissa, and delay margin.

\begin{example}

We consider a single-input single-output linear time-invariant system, given by:
\begin{equation}
\label{eq:example_1}
    \Sigma := 
    \left\{
    \begin{array}{l}
    \dot{x}(t) =
    \begin{bmatrix}
    0 & 1 & 0 \\
    0 & 0 & 1 \\
    -2 & -3 & -1
    \end{bmatrix} x(t)
    +
    \begin{bmatrix}
    0 \\ 0 \\ 1
    \end{bmatrix} u(t-\tau_0), \\[3mm]
    y(t) = [\,1 \;\; 0 \;\; -0.5\,]\, x(t),
    \end{array}
    \right.
\end{equation}
where $x(t)\in\mathbb{R}^3$, $u(t)\in\mathbb{R}$, and $\tau_0 \geq 0$ denotes a constant input delay.
The system is open-loop stable and admits the transfer function representation:
\begin{equation}
    H_{yu}(s) := \frac{1 - 0.5s^2}{s^3 + s^2 + 3s + 2}e^{-\tau_0s}.
\end{equation}
Although stable, the system exhibits non-minimum-phase behavior and slow dynamics, with open-loop poles located at $s = -0.7152$ and $s = -0.1423 \pm 1.666\,\boldsymbol{i}$.

The control objective is to design a filtered PID controller of the form
\begin{equation}
\label{eq:example_pid}
    \left\{
    \begin{array}{l}
    u(t) = k_p y(t) + k_i \displaystyle\int_0^t y(s)\,ds + k_d z(t), \\[3mm]
    T \dot{z}(t) + z(t) = \dot{y}(t),
    \end{array}
    \right.
\end{equation}
such that the spectral abscissa of the closed-loop system including the filter is minimized.
With the ideal derivative action, the corresponding closed-loop characteristic function $\Delta : \mathbb{C} \to \mathbb{C}$ is given by
\begin{equation}
    \Delta(s) = s^4 + s^3 + 3s^2 + 2s + e^{-\tau_0s}(1 - 0.5s^2)\bigl(k_d s^2 + k_p s + k_i\bigr).
\end{equation}

\subsubsection*{Classical design without arbitrary small $T$.}
We first minimize the spectral abscissa of the closed-loop system without considering the inclusion of the derivative filter. Using a BFGS-based optimization algorithm, a local minimizer is obtained with
\[
    k_p = 0.6439, \qquad
    k_d = 1.9, \qquad
    k_i = 0.4222.
\]
These gains place the dominant closed-loop roots at
$s = -0.3004 \pm 0.0898\,\boldsymbol{i}$ when $\tau_0 = 0$.
However, since $|k_d| > 1$, the resulting closed-loop system becomes neutral and becomes unstable for any input delay $\tau_0 > 0$. As a consequence, the inclusion of a low-pass filter in the derivative action is necessary.

Choosing a small filter constant, e.g.\ $T = 10^{-3}$, preserves approximately the same spectral abscissa in the delay-free case. The corresponding stability region in the $(T,\tau_0)$-plane is shown in Figure~\ref{fig:stableRegionClassicEx}. It can be observed that, for this choice of parameters, the admissible delay margin is extremely small, with $\tau_{\min} < 7 \times 10^{-3}$.

\begin{figure}[t]
\centering
    \includegraphics[width=0.5\linewidth]{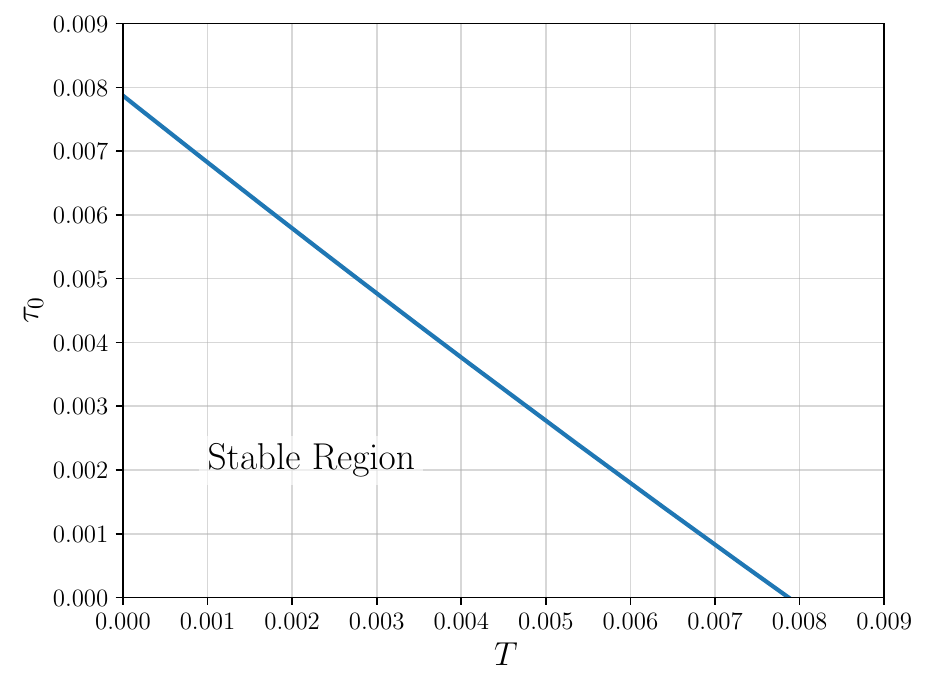}
    \caption{Stability region in the $(T,\tau_0)$-plane for the controller obtained using the classical design approach.}
    \label{fig:stableRegionClassicEx}
\end{figure}

This behavior may be due to the fact that $k_d>1$, which makes high-frequency roots more sensitive, and, as consequence, the controller less robust. 

\subsubsection*{Proposed delay-robust design.}
We now apply the filtered PID design procedure described in Algorithm \ref{alg:filteredPID}. The objective is to minimize the spectral abscissa of the filtered system, while aiming to increase the delay margin. 

When incorporating the filter into the design procedure, it is often advantageous to initialize the optimization with a relatively large value of $T$, allowing the algorithm to explore regions associated with improved robustness. In particular, we consider the initial values $k_p = k_d = k_i = T = 0.1$, while restricting $T \in[0.001,0.2]$. Although this choice is arbitrary, it satisfies the minimal requirement of the algorithm, namely that the initial gains stabilize the delay-free system.

Similarly to the previous approach, using a BFGS optimization method we obtain a local minimizer with the following controller parameters:
\[
    k_p = 0.0360, \qquad
    k_i = 0.1103, \qquad
    k_d = 0.5005, \qquad
    T = 0.1543 .
\]

In the delay-free case, these values place the dominant closed-loop roots at
\[
    s = -0.1011 \pm 1.6262\,\boldsymbol{i},
\]
which correspond to a slower convergence rate than that obtained with the classical design. Nevertheless, the resulting closed-loop system admits a delay margin of approximately $\tau_{\min} \approx 0.85$, representing an improvement of nearly two orders of magnitude compared with the classical approach.

At first glance, it may appear counterintuitive that a larger value of $T$ leads to an improved delay margin. Indeed, the filter can be interpreted as introducing dynamics that resemble a delay, which would suggest that increasing $T$ should reduce robustness. However, a larger value of $T$ attenuates high-frequency dynamics and shifts the corresponding characteristic roots further into the left half-plane. As a result, although the closed-loop response becomes slower, the system exhibits significantly improved robustness with respect to input delays.

The associated stability region in the $(T,\tau_0)$-plane is shown in Figure ~\ref{fig:stabFiltered}. This substantial increase in delay robustness is achieved by explicitly incorporating the filter constant $T$ into the optimization process.

\begin{figure}
\centering
    \includegraphics[width=0.5\linewidth]{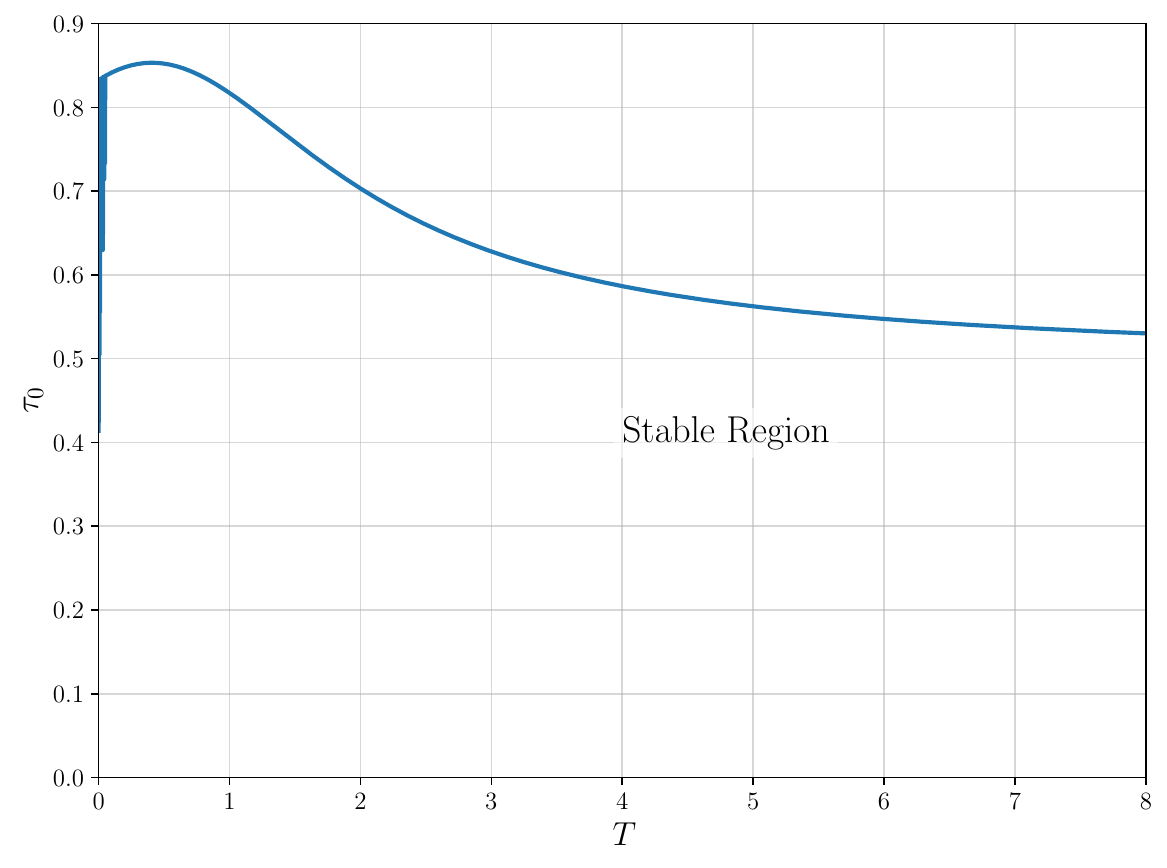}
    \caption{Stability region in the $(T,\tau_0)$-plane for the controller obtained using Algorithm~\ref{alg:filteredPID}. The shaded region indicates combinations of filter constant $T$ and input delay $\tau_0$ for which the closed-loop system remains exponentially stable.}
    \label{fig:stabFiltered}
\end{figure}

To better understand how the stability region is modified and how this impacts the robustness of the closed-loop system, Figs.~\ref{fig:roots_delay} and~\ref{fig:roots_filter} illustrate the evolution of the rightmost characteristic roots under both controller configurations.

More precisely, Figure ~\ref{fig:roots_delay} shows the movement of the dominant roots as the input delay $\tau_0$ varies, while Figure ~\ref{fig:roots_filter} depicts their dependence on the filter constant $T$. These plots highlight the different sensitivities of the closed-loop spectrum with respect to delay and filtering parameters under the classical and filtered design strategies.

\begin{figure}
    \centering
    \includegraphics[width=0.8\linewidth]{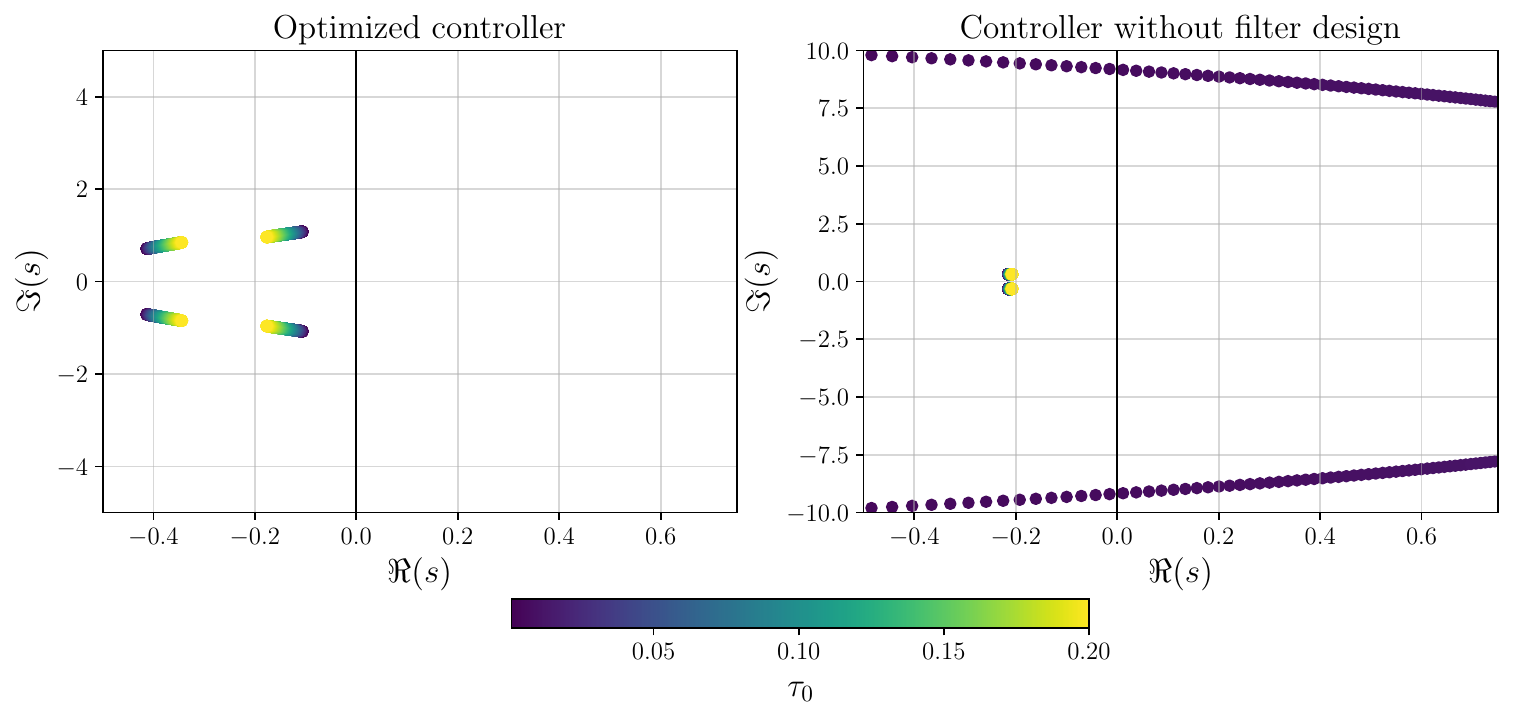}
    \caption{Evolution of the rightmost characteristic roots as the input delay $\tau_0$ varies for the classical and filtered included PID designs. The filtered design exhibits significantly improved robustness with respect to delay variations.}
    \label{fig:roots_delay}
\end{figure}

\begin{figure}
    \centering
    \includegraphics[width=0.8\linewidth]{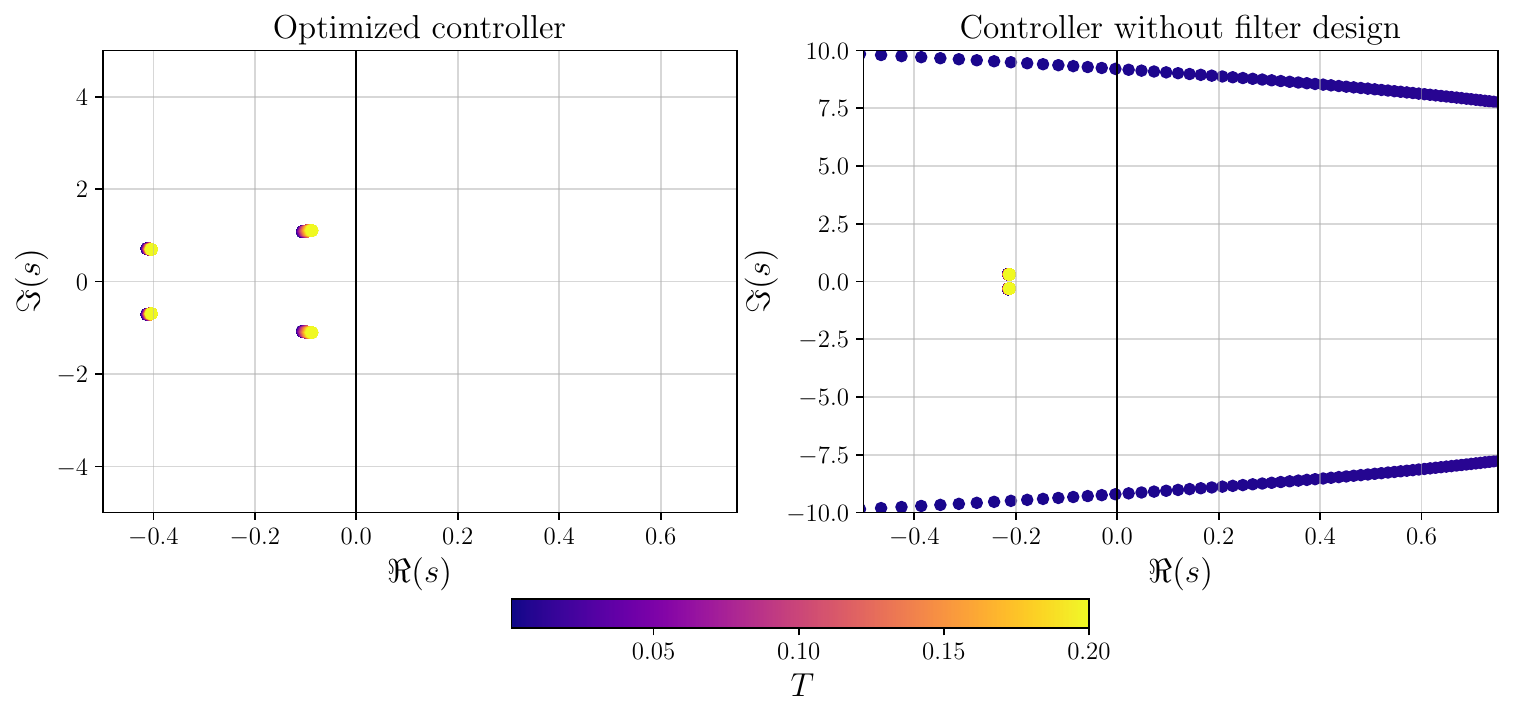}
    \caption{Evolution of the rightmost characteristic roots as the filter constant $T$ varies. The figure illustrates the sensitivity of the closed-loop spectrum to the filter dynamics.}
    \label{fig:roots_filter}
\end{figure}

\end{example}

The previous example highlights the importance of incorporating the filter constant into the control design. In the following example, we illustrate a practical application of this co-design approach. In particular, we have observed that the delay margin of the closed-loop system can be significantly increased through an appropriate choice of the design parameters. This naturally motivates the use of the proposed methodology to tune the controller gains and filter constant such that robustness is increased, rather than just minimizing the spectral abscissa.

\begin{example}
Delays are ubiquitous in control engineering systems, and enhancing the delay margin of a closed-loop system is therefore a key objective in controller design. However, the explicit computation of the delay margin is, in general, a challenging task (see, e.g., \cite{Ma2018delay, Chen2021delay, Ramirez2021scalable}), and numerical methods are typically required in practice.

With this in mind, consider the following single-input single-output linear time-invariant system:
\begin{equation}
\Sigma :=
\left\{
    \begin{array}{l}
    \dot{x}(t) =
    \begin{bmatrix}
    0 & 1  \\
    1 & 0
    \end{bmatrix} x(t)
    +
    \begin{bmatrix}
    0 \\ 1
    \end{bmatrix} u(t-\tau_0), \\
    y(t) = [-2 \;\; 1] x(t),
    \end{array}
\right.\label{eq:example_illustrative}
\end{equation}
where $x(t)\in\mathbb{R}^2$, $u(t)\in\mathbb{R}$, and $\tau_0 \geq 0$ denotes a constant input delay.

In the delay-free case ($\tau_0 = 0$), the system cannot be stabilized using proportional feedback alone, nor by means of a proportional–integral (PI) controller. For this reason, we consider a proportional–derivative (PD) control law of the form
\[
\left\{
    \begin{array}{l}
    u(t) = k_p y(t) + k_d z(t), \\
    T \dot{z}(t) + z(t) = \dot{y}(t),
    \end{array}
\right.
\]
where $z(t)$ is a filtered approximation of the output derivative and $T>0$ is the filter time constant.

The use of a filtered derivative is essential, since an ideal derivative action would lead to a neutral-type closed-loop system. In the delay-free case, the closed-loop system is exponentially stable whenever
\[
    -2 < k_p < -\frac{1}{2}, \qquad -1 < k_d < \frac{k_p}{2}.
\]

When an input delay $\tau_0 > 0$ is present, the use of an ideal derivative would typically impose the additional stability constraint $|k_d| < 1$. This restriction, however, can be relaxed when the derivative action is implemented through the low-pass filter introduced above.

The objective is to minimize the spectral abscissa of the closed-loop system by selecting appropriate control gains and filter constant. Algorithm \ref{alg:filteredPID} requires stabilizing parameters as an initial condition. Consider, for instance, the choice $k_p = -1$, $k_d = -5/6$, and $T = 0.01$, with an admissible interval $T \in [0.001, 0.06]$. For these parameters, the system is stable and can tolerate a nonzero delay. The corresponding stability region in the $(T, \tau_0)$-plane is depicted in Figure \ref{fig:first_region_illustrative}.

\begin{figure}
    \centering
    \includegraphics[width=0.5\linewidth]{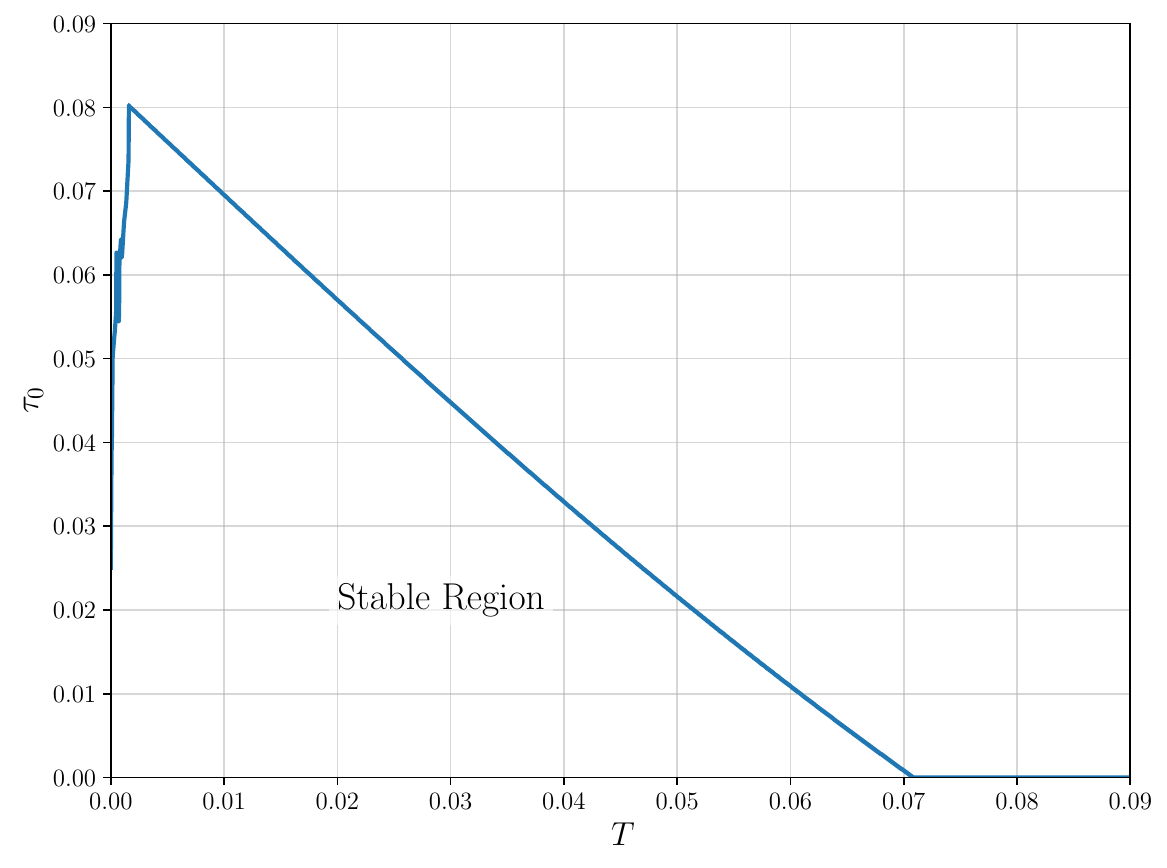}
    \caption{Stability region in the $(T, \tau_0)$-plane for the initial set of parameters.}
    \label{fig:first_region_illustrative}
\end{figure}

Applying Algorithm \ref{alg:filteredPID} with a BFGS scheme yields a local minimum at $k_p = -1.2311$, $k_d = -0.8927$, and $T = 0.0059$. The corresponding stability region is shown in Figure \ref{fig:second_region_illustrative}. Although the spectral abscissa is improved, the system becomes less robust: a smaller value of $T$ increases sensitivity to high-frequency dynamics. Moreover, the stability region shrinks, resulting in a more fragile closed-loop behavior.

\begin{figure}
    \centering
    \includegraphics[width=0.5\linewidth]{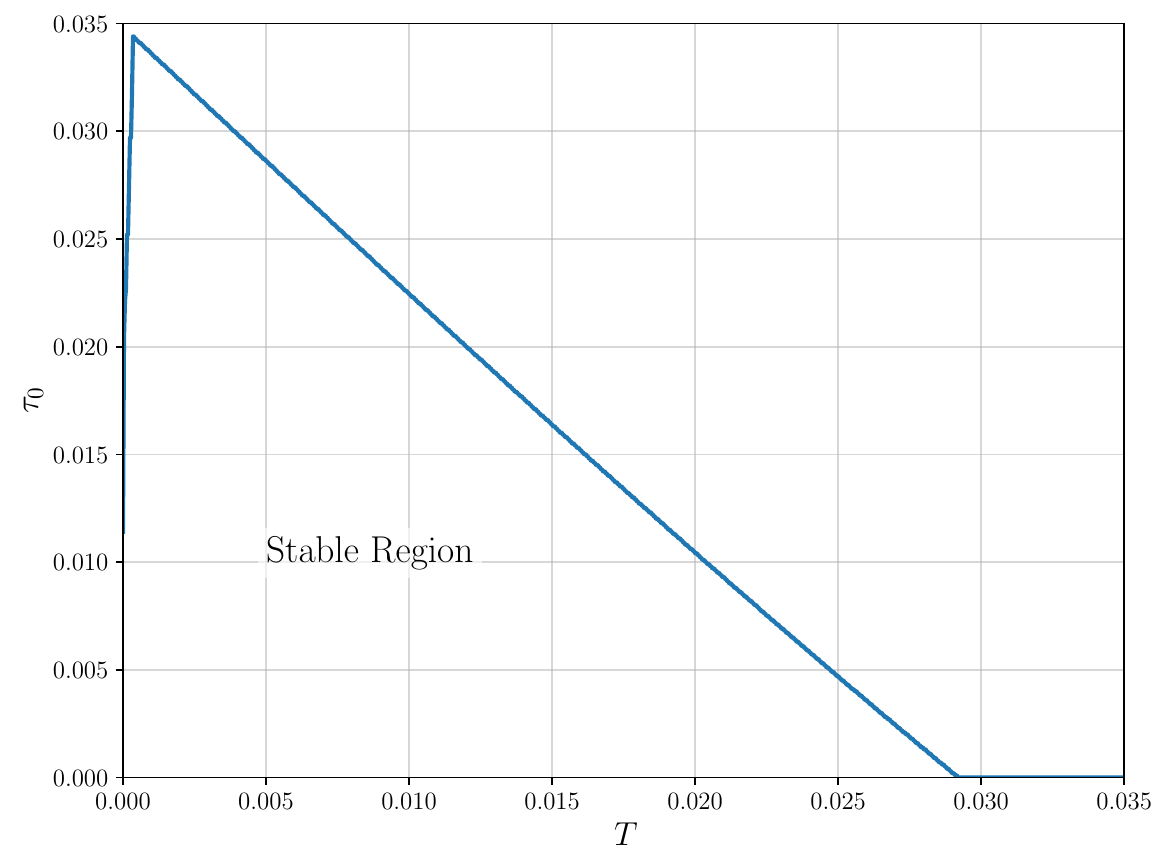}
    \caption{Stability region corresponding to the locally optimal parameters obtained from Algorithm \ref{alg:filteredPID}.}
    \label{fig:second_region_illustrative}
\end{figure}

Based on Figure \ref{fig:first_region_illustrative}, one may impose bounds on $T$ to improve robustness. In this spirit, we repeat Algorithm \ref{alg:filteredPID} with the same initial gains, but restrict $T \in [0.02, 0.04]$, using $T=0.02$ as the initial value. This choice aims to preserve a reasonable delay margin while reducing high-frequency sensitivity. The resulting local minimum is $k_p = -0.8184$, $k_d = -0.7237$, and $T = 0.0216$. The corresponding stability region, shown in Figure \ref{fig:third_region_illustrative}, is significantly larger. For these parameters, the system can tolerate delays up to approximately $\tau_0 \approx 0.13$, while also benefiting from improved high-frequency behavior due to the larger value of $T$.

\begin{figure}
    \centering
    \includegraphics[width=0.5\linewidth]{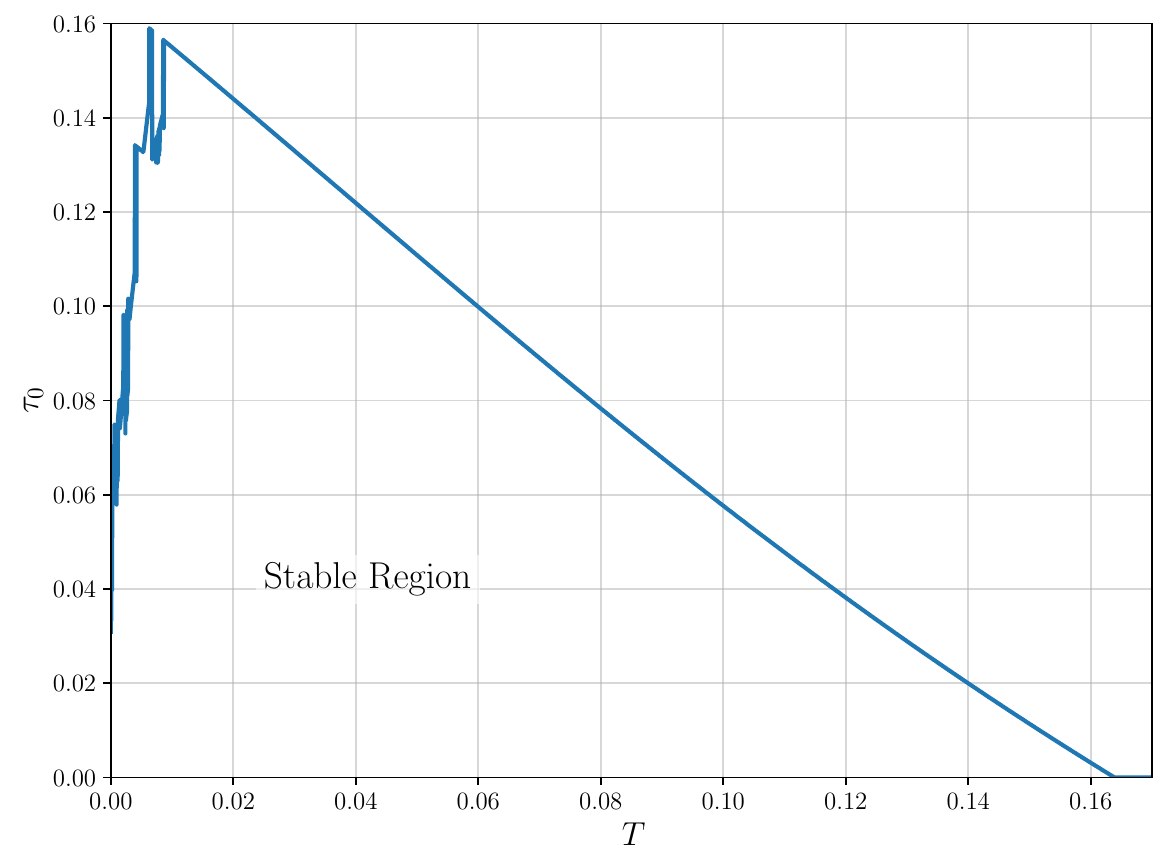}
    \caption{Stability region obtained after a second optimization with a restricted interval for $T$.}
    \label{fig:third_region_illustrative}
\end{figure}

In a final iteration, we consider $T \in [0.06, 0.1]$. This yields $k_p = -0.6618$, $k_d = -0.5999$, and $T = 0.0729$. As shown in Figure \ref{fig:fourth_region_illustrative}, the corresponding delay margin increases to approximately $\tau_{\text{max}} \approx 0.2$. 

\begin{figure}
    \centering
    \includegraphics[width=0.5\linewidth]{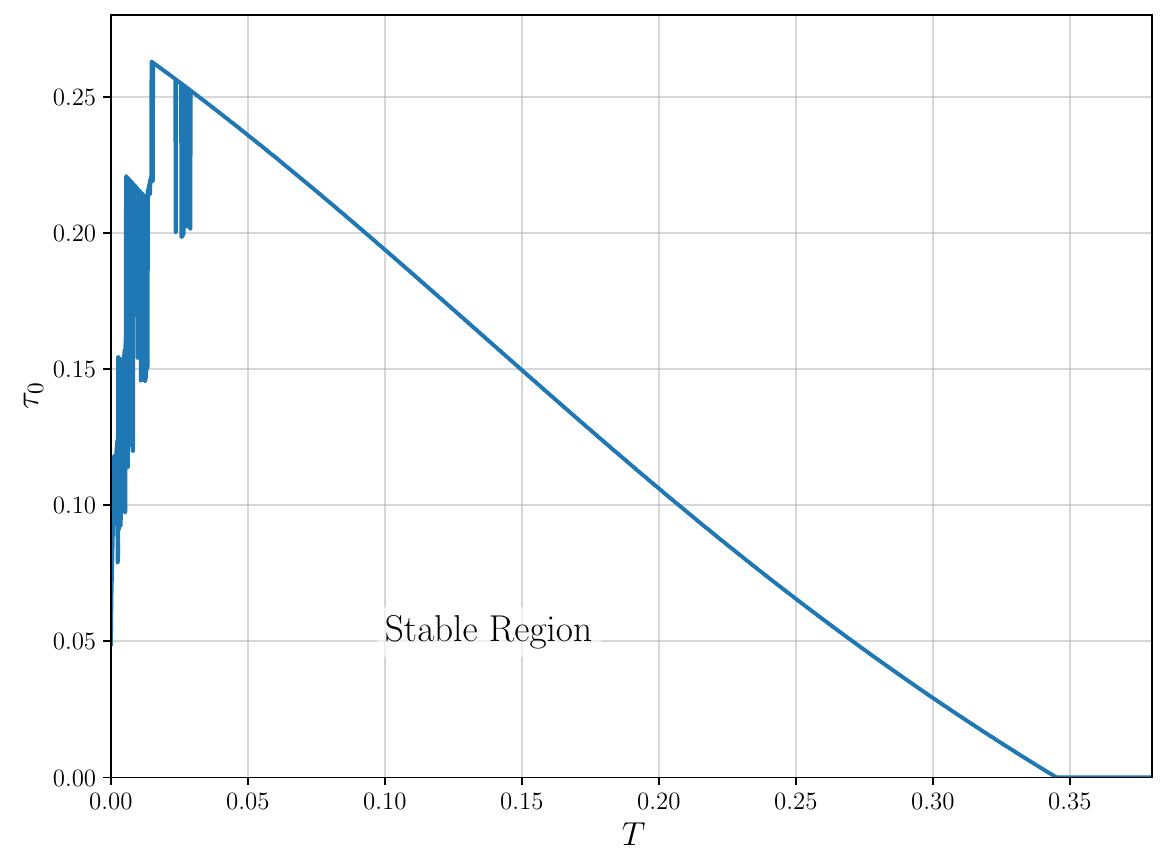}
    \caption{Stability region obtained after a third iteration of Algorithm \ref{alg:filteredPID}.}
    \label{fig:fourth_region_illustrative}
\end{figure}

We summarize the delay-free spectral abscissa and delay margins corresponding to the different controller gains in Table \ref{tab:results}. It is worth noting that, on the first set of \emph{optimized} gains, the delay margin is being highly influenced by the high-frequency roots that increase weight when $T$ goes to zero.

\begin{table}[h!]
\centering
\caption{Spectral abscissa and delay margin for each controller}
    \label{tab:results}
    \begin{tabular}{|c|c|c|c|}
    \hline
    \textbf{Control gains $(k_p, k_d)$} & \textbf{Filter constant $T$} & \textbf{Spectral abscissa $\rho$} & \textbf{Delay margin $\tau_\text{max}$} \\
    \hline
    % Example row:
    $(-1.2311, -0.8927)$ & $0.0059$ & $-3.57769$ & $0.0275$ \\
    \hline
    $(-0.8184, -0.7237)$ & $0.0216$ & $-1.46503$ & $0.1422$ \\
    \hline
    $(-0.6618, -0.5999)$ & $0.0729$ & $-1.23445$ & $0.217$ \\
    \hline
    % Add more rows here
\end{tabular}
\end{table}

This example illustrates how improved robustness can be achieved, while still reducing the spectral abscissa, through careful tuning of both the controller gains and the filter constant $T$. Although the focus here is on enhancing the delay margin, the same heuristic approach can be adapted to improve other robustness properties, depending on the application.
\end{example}

\section{Concluding remarks}

In this paper, we have addressed the design of filtered PID controllers for linear time-invariant systems with delays by explicitly accounting for the impact of the derivative filter on the closed-loop spectral properties. Motivated by the observation that classical PID design approaches based on ideal derivative action may lead to poor delay robustness or even loss of stability when filtering is introduced, we proposed a co-design framework in which the controller gains and the filter constant are optimized simultaneously.

The core contribution of the proposed framework is the formulation of a spectral-abscissa minimization problem subject to an explicit interval constraint on the filter constant. This problem constitutes the fundamental optimization step of the method and provides a systematic means of tuning the controller gains together with the filter parameter within a prescribed admissible range.

The numerical examples illustrates that the proposed strategy can significantly increase, for example, the achievable delay margin compared to classical approaches, at the expense of a moderate reduction in nominal convergence rate. This trade-off is intrinsic to delay-robust control design and highlights the importance of explicitly incorporating robustness considerations at the optimization stage.

\bibliographystyle{abbrv}
\bibliography{refs}
\end{document}